\documentclass[submission,copyright,creativecommons]{eptcs}
\providecommand{\event}{Refine 2011} 

\usepackage{color}
\usepackage{url}
\usepackage[latin1]{inputenc}
\usepackage{amsmath}
\usepackage{amscd}
\usepackage{latexsym}
\usepackage{stmaryrd}
\usepackage{amssymb}
\usepackage{amscd}
\usepackage{fancyvrb}
\usepackage{boxedminipage}

\ifpdf
\usepackage[pdftex]{graphicx}
\else
\usepackage{graphicx}
\fi
 
 \usepackage[all]{xy}

 \def\g{\fuc{G}}
 \usepackage{cjr}

\def\setdef#1#2{\{#1|\;#2\}}             
\def\enset#1{\mathopen{ \{ }#1\mathclose{ \} }} 
\def\N{\blk{N}}
\def\fuc#1{\mathsf{#1}}
\def\dimp{\mathbin{\Leftrightarrow}}
\def\imp{\mathbin{\Rightarrow}}

\def\e{\mathbin{\wedge}}

\def\ie{{\em i.e.\/}}

\def\eg{{\em e.g.\/}}
\def\pair#1{\const{#1}}
\def\const#1{\mathopen{\langle}#1\mathclose{\rangle}}
\newtheorem{definition}{Definition}
\newtheorem{lemma}{Lemma}
\newtheorem{theorem}{Theorem}
\newtheorem{corollary}{Corollary}
\newtheorem{example}{Example}

\newenvironment{proof}%
     {\vspace{1mm}\noindent \textbf{Proof.}}%
    {\par \nopagebreak \noindent \qed \vspace{3mm} }

\title{Refinement by interpretation in $\pi$-institutions}
\author{César J.~Rodrigues
\institute{Dep. Informatics \& CCTC,  \\ Minho University, Portugal}
\email{cjr@di.uminho.pt}
\and
Manuel A. Martins 
\institute{Dep. Mathematics, \\ Aveiro University,  Portugal} 
\email{martins@ua.pt} 
\and
Alexandre Madeira 
\institute{Dep. Informatics \& CCTC,  \\ Minho University, Portugal\\ Dep. Mathematics, \\ Aveiro University,  Portugal\\
Critical Software S.A., Portugal} 
\email{madeira@ua.pt} 
\and
Luís S.~Barbosa
\institute{Dep. Informatics \& CCTC,  \\ Minho University, Portugal}
\email{lsb@di.uminho.pt}
}
\def\titlerunning{Refinement by interpretation in $\pi$-institutions}
\def\authorrunning{C.J. Rodrigues, M.A. Martins, A. Madeira \& L.S. Barbosa}
\begin{document}
\maketitle

\begin{abstract}
The paper discusses the role of interpretations, understood as multifunctions that
preserve and reflect logical consequence, as refinement witnesses in the
general setting of $\pi$-institutions. This leads to a smooth generalization of the
``refinement by interpretation" approach, recently introduced by the authors
 in more specific contexts. As a second, yet related contribution
  a basis is provided to build up a refinement  calculus of 
structured specifications in and across arbitrary $\pi$-institutions.
\end{abstract}

\section{Introduction}\label{sc:in}

The expression \emph{refinement by interpretation} was coined in \cite{MMB09a} to refer to 
an alternative approach to refinement
of equational specifications in which signature morphisms are replaced by 
 \emph{logical interpretations} as refinement witnesses.  
 
 Intuitively, an interpretation is a logic translation which preserves and reflects
 meaning. 
  Actually, it is a central tool in the study of
equivalent algebraic semantics (see, \eg,
\cite{Woji,BP89,BP,BR03,proto}),  a paradigmatic example being the interpretation of the
\emph{classical propositional calculus} into the \emph{equational theory of boolean algebras}
(cf. \cite[Example 4.1.2]{BP}).
 Interestingly enough, and in the more operational setting of formal software development, 
 the notion of interpretation proved effective to capture a number of transformations difficult to deal with in classical terms. 
 Examples include data encapsulation and the decomposition of operations into atomic transactions \cite{MMB09a}.
 
 A typical refinement pattern that is not easily captured by the classical approach concerns 
 refinement of a  subset of operations into operations defined over more specialized sorts.
 This kind of transformation induces the loss of the functional property on the operations' component of  signature morphisms. For example, there is not a  signature morphism $\sigma$ to guide a refinement where a specification with operations $g:s'\rightarrow s$ and $f:s'\rightarrow s$ is transformed into one with operations $g:s' \rightarrow s_{new}$ and $f:s'\rightarrow s$, since this translation naturally induces a map $\sigma_{sort}(s)=\{s,s_{new}\}$ which  violates the definition of signature morphism.

The approach seems also promising in the context of new, emerging computing paradigms which entail the 
  need for more flexible approaches to what is taken as a valid transformation of specifications, as in, for example, \cite{batory}.
Later, in \cite{MMB09b}, the whole framework  was generalized from the original equational setting to 
address deductive systems of arbitrary dimension. This made possible, for
example, to refine sentential into equational specifications and the latter into modal  ones.
Moreover, the restriction to logics with finite consequence relations was dropped which resulted in
increased flexibility along the software development process. The interested reader is referred to both papers for a
number of illustrative examples.

On the other hand, the notion of an institution \cite{instituicoes}, proposed by J. Goguen and R. Burstall in the late 1970s, 
has proven very successful in formalizing logical systems and their interrelations.

This paper aims at lifting the use of logic interpretations to witness refinement of specifications at an institutional
level. This is made in the context of $\pi$-institutions \cite{FS88} which deal directly with
syntactic consequence relations rather than with semantical satisfaction, as in the original
definition of an institution \cite{instituicoes}.  $\pi$-institutions are particularly useful in 
formalizing deductive systems with varying signatures, which
are only indirectly handled by the methods of abstract algebraic logic, as in \cite{BP} on which our first 
generalization \cite{MMB09b} is based. In general, $\pi$-institutions
provide a more operational framework with no loss of 
expressiveness as any classical institution can be suitably translated.

Refinement by interpretation is proposed here at two different levels: a \emph{macro} level relating
different $\pi$-institutions, and the \emph{micro} level of specifications inside a particular, although arbitrary, $\pi$-institution.
The former discusses what is an interpretation of  institutions and provides the envisaged generalization of this 
approach to refinement of arbitrary deductive systems. The latter, on the other hand, corresponds to a sort of
\emph{local} refinement witnessed by interpretations thought simply as multifunctions relating sentences generated
by different signatures within the same institution.

As a second, although related, contribution, the paper
lays the basis for a re\-fi\-nement-by-interpretation calculus of structured specifications in an arbitrary  (and across)
$\pi$-institution(s).
That both levels can be addressed and related to each other  comes to no surprise:
 a main outcome of institution theory is precisely to provide
what   \cite{AN94} describes as
\emph{
effective mechanisms to manipulate theories in
an ana\-lo\-gous way as our deductive calculi manipulate formulas}.

The remainder of this paper is organized as follows. $\pi$-institutions and a notion of interpretation between them
are reviewed in section \ref{sc:pi}. Then, section  \ref{sc:rbi} characterizes refinement by interpretation in this context,
whereas the local view is discussed in section \ref{sc:inside}. The structure of a refinement calculus is discussed in
 section \ref{sc:calc}. Section \ref{sc:conc} concludes and
highlights some pointers to related work.

\section{$\pi$-institutions and interpretations}\label{sc:pi}

In broad terms, an institution consists of an arbitrary category $Sign$ of signatures together with two functors 
$\fuc{SEN}$ and $\fuc{MOD}$ that give, respectively, for each signature, a set of sentences and a category of models.
For each signature, sentences and models are related via a satisfaction relation whose  main axiom formalizes the 
popular aphorism \emph{truth is invariant under change of notation} \cite{livrodiaconescu}. Such a very generic way to 
capture a logical system was originally motivated by quite pragmatic concerns: to provide 
an abstract, language-independent framework for specificifying and reasoning about software systems, in response to the explosion
of specification logics. Several current specification formalisms, notably, \textsc{CafeOBJ} \cite{DF02}, \textsc{Casl} \cite{MHST03}
and \textsc{Hets} \cite{tacas} were designed to take advantage of such a general framework.

$\pi$-institutions, proposed by J. Fiadeiro and A. Sernadas in \cite{FS88}, fulfill a si\-mi\-lar role, replacing  semantical satisfaction
by  a syntactic consequence relation \emph{à la} Tarski. Therefore,  a $\pi$-institution introduces, for each 
signature, a closure operator on the set of its sentences capturing logical consequence. As remarked by
G. Voutsadakis in \cite{Vo03} $\pi$-institutions \emph{may be viewed as the natural generalization of the notion of a deductive 
system on which a categorical theory of algebraizability, generalizing the theory of \cite{BP} may be based}. In the sequel
we review the basic definition and adopt  Voutsadakis's notion of interpretation to define refinement by interpretation in such
a general setting.

\begin{definition}
A $\pi$-institution $I$ is a tuple
$\pair{Sign,\fuc{SEN},(C_{\Sigma})_{\Sigma \in |Sign|}}$ where
\begin{itemize}
\item $Sign$ is a  category of signatures and signature morphisms;
\item $\fuc{SEN}: Sign\rightarrow Set$ is a functor  from the category of
signatures to the category of small sets giving, for each
$\Sigma \in |Sign|$, the set $\fuc{SEN}(\Sigma)$ of $\Sigma$-sentences
 and mapping each $f:\Sigma_{1}\rightarrow \Sigma_{2}$
to a {\em substitution} $\fuc{SEN}(f): \fuc{SEN}(\Sigma_{1})\rightarrow \fuc{SEN}(\Sigma_{2})$;
\item for each $\Sigma \in |Sign|$, $C_{\Sigma}: \pow{(\fuc{SEN}({\Sigma}))}\rightarrow
\pow{(\fuc{SEN}({\Sigma}))}$ is a mapping, called
$\Sigma$-{\em closure}, such that, for all $A,B\subseteq \fuc{SEN}(\Sigma)$ and $\Sigma_1, \Sigma_2 \in Sign$;
\begin{description}
\item[(a)] $A\subseteq C_{\Sigma}(A)$
\item[(b)] $C_{\Sigma}(C_{\Sigma}(A))= C_{\Sigma}(A)$
\item[(c)] $C_{\Sigma}(A)\subseteq  C_{\Sigma}(B)$
for  $A\subseteq B$
\item[(d)] $\fuc{SEN}(f)(C_{\Sigma_{1}}(A)) 
\subseteq
C_{\Sigma_{2}}(\fuc{SEN}(f)(A))$
\end{description}
\end{itemize}
\end{definition}
Note that the $\Sigma$-closure operator of a $\pi$-institution is not required to be finitary.
 
 \begin{definition}
 A $\pi$-institution
$I^{\prime} = \pair{Sign^{\prime},\fuc{SEN}^{\prime},(C^{\prime}_{\Sigma})_{\Sigma \in |Sign^{\prime}|}}$ is a 
\emph{sub-$\pi$-institution} of 
$I = \pair{Sign,\fuc{SEN}, \\ (C_{\Sigma})_{\Sigma \in |Sign|}}$ if $Sign'$ is a sub-category of $Sign$ and,
for each $\Sigma \in |Sign'|$, $\fuc{SEN'}(\Sigma) \subseteq \fuc{SEN}(\Sigma)$
and the $\Sigma$-closure $C'_{\Sigma}$ is the restriction of $C_{\Sigma}$.
\end{definition}

 Roughly speaking, the notion of logical interpretation underlying \cite{MMB09b} is that of \cite{BP89}: a multifunction
 (i.e., a set-valued function) relating formulas which preserves and reflects logical consequence. 
 Note that the expressive flexibility of interpretations comes precisely from their definition as  multifunctions.
  A corresponding definition, to be used in the sequel, was proposed, in the context of $\pi$-institutions,
  in \cite{Vo03}:
  
\begin{definition}
  Given two $\pi$-institutions  $I  = \pair{Sign,\fuc{SEN},(C_{\Sigma})_{\Sigma \in |Sign|}}$ and 
  $I^{\prime}  = \pair{Sign^{\prime},\fuc{SEN}^{\prime},(C^{\prime}_{\Sigma})_{\Sigma \in |Sign^{\prime}|}}$, a 
{\em translation} $\pair{\f,\alpha}:I\rightarrow I^{\prime}$ consists
of a functor $\f:Sign\rightarrow Sign^{\prime}$ together with a
natural transformation  $\alpha: \fuc{SEN}\rightarrow \pow{~}\fuc{SEN}^{\prime} \f$.

A translation 
$\pair{\f,\alpha}:I\rightarrow I^{\prime}$ is a {\em semi-interpretation}
if, for all $\Sigma \in |Sign|$, $\Phi\cup \enset{\phi}\subseteq 
\fuc{SEN}(\Sigma)$,
\begin{eqnarray}
\phi \in C_{\Sigma}(\Phi) &\; \; \imp\; \; &
\alpha_{\Sigma}(\phi)\subseteq C^{\prime}_{\f(\Sigma)}(\alpha_{\Sigma}(\Phi))
\label{semii}
\end{eqnarray}
It is an {\em interpretation} if, 
\begin{eqnarray}
\phi \in C_{\Sigma}(\Phi) & \; \; \dimp\; \; &
\alpha_{\Sigma}(\phi)\subseteq C^{\prime}_{\f(\Sigma)}(\alpha_{\Sigma}(\Phi))
\end{eqnarray}
Finally, we say that a translation $\pair{\f,\alpha}$ interprets a $\pi$-institution $I$, if there is 
a $\pi$-institution $I^0 = \pair{Sign^0,\fuc{SEN}^0,(C^0_{\Sigma})_{\Sigma \in |Sign^0|}}$ for which 
$\pair{\f,\alpha}$ is an interpretation.
\end{definition}

Note that a translation depends only on the categories of signatures and 
the sentence functors involved, but  not on the family of closure operators.
A translation is a {\em self-translation} if  $\f$ is the identity functor $Id$.
On the other hand, it is said to be a {\em functional translation} if, for every 
$\Sigma \in |Sign|$, $\phi \in \fuc{SEN}(\Sigma)$,
$|\alpha_{\Sigma}(\phi)|=1$. Additionally, it is  an
{\em identity translation},  if for every
$\Sigma \in |Sign|$, $\phi \in \fuc{SEN}(\Sigma)$,
\begin{eqnarray}
\alpha_{\Sigma}(\phi)=\enset{\phi} \label{sitdef} 
\end{eqnarray}

 \section{Refining $\pi$-institutions by interpretation}\label{sc:rbi}

In software development the process of \emph{stepwise refinement} \cite{ST88x} 
encompasses a chain of successive transformations of a specification
\[S_{0}\leadsto S_{1}\leadsto S_{2}\leadsto \cdots
\leadsto S_{n-1}\leadsto S_{n} \]
through which
 a complex design is produced by incrementally adding details and reducing under-spe\-ci\-fi\-ca\-tion.
This is done step-by-step until the class of models becomes restricted to such an extent that a
program can be easily manufactured.
The discussion on what counts for a valid refinement step, represented by 
$S_{i}\leadsto S_{j}$, is precisely the starting point of this line of research.

The minimal requirement to be placed on a refinement relation, besides being a pre-order to allow stepwise construction, is 
preservation of logical consequence. In the framework of $\pi$-institutions this corresponds to the following definition:

\begin{definition}[Syntactic refinement]
Let   $I  = \pair{Sign,\fuc{SEN},(C_{\Sigma})_{\Sigma \in |Sign|}}$ and 
  $I^{\prime}  = \pair{Sign^{\prime},\fuc{SEN}^{\prime},(C^{\prime}_{\Sigma})_{\Sigma \in |Sign^{\prime}|}}$
be two $\pi$-institutions. $I^{\prime}$ 
is  a syntactic refinement of $I$ if $Sign$ is a sub-category of $Sign^{\prime}$ and,
for each $\Sigma \in |Sign|$,
 $\fuc{SEN}(\Sigma) \subseteq \fuc{SEN}^{\prime}(\Sigma)$ and
$C_{\Sigma}(\Phi)\subseteq C^{\prime}_{\Sigma}(\Phi)$ 
for $\Phi\subseteq \fuc{SEN'}(\Sigma)$.
\end{definition}
Clearly, a  $\pi$-institution is  a syntactic refinement of any of its $\pi$-sub-institutions.
Refinement by interpretation, on the other hand, goes a step further:

\begin{definition}[Refinement by interpretation]\label{df:rbi}
Consider two $\pi$-institutions
 $I  = \pair{Sign,\fuc{SEN},(C_{\Sigma})_{\Sigma \in |Sign|}}$ and 
  $I^{\prime}  = \pair{Sign^{\prime},\fuc{SEN}^{\prime},(C^{\prime}_{\Sigma})_{\Sigma \in |Sign^{\prime}|}}$ and let
 $\fdec{\pair{\f,\alpha}}{I}{I'}$  be a translation.
$I'$  is a \emph{refinement by interpretation} of $I$ via $\pair{\f,\alpha}$, written as $I \leadsto_{\pair{\f,\alpha}} I'$, if
\begin{itemize}
\item there is a $\pi$-institution $I^0 = \pair{Sign^{\prime}, \fuc{SEN}^{\prime}, (C^0_{\Sigma})_{\Sigma \in Sign^{\prime}}}$
that interprets $I$ under translation $\pair{\f,\alpha}$;
\item for all 
$\Sigma \in |Sign|$, $\Phi\subseteq \fuc{SEN}(\Sigma)$,
  \[
\phi \in C_{\Sigma}(\Phi)\; \; \imp\; \;
\alpha_{\Sigma}(\phi)\subseteq C^{\prime}_{\f(\Sigma)}(\alpha_{\Sigma}(\Phi))
\]
\end{itemize}
\end{definition}
Clearly, a syntactic refinement is a refinement by interpretation for a self, identity, functional interpretation,
with $\f = Id$.
The following Lemma establishes  an useful characterization of refinement via interpretation:
\begin{lemma}
Let $I  = \pair{Sign,\fuc{SEN},(C_{\Sigma})_{\Sigma \in |Sign|}}$ and 
  $I^{\prime}  = \pair{Sign^{\prime},\fuc{SEN}^{\prime},(C_{\Sigma})_{\Sigma \in |Sign^{\prime}|}}$
 be two $\pi$-institutions  and $\fdec{\pair{\f,\alpha}}{I}{I'}$  a translation. 
 Then, $I \leadsto_{\pair{\f,\alpha}} I'$ if $I'$ is a syntactic refinement of some interpretation of $I$ through  $\langle \f,\alpha\rangle$. 
\end{lemma}
\begin{proof}
Suppose $I^{\prime}$ is a syntactic refinement of an arbitrary interpretation $I^0$ of $I$ along $\pair{\f,\alpha}$. 
Clearly the first condition in the definition of refinement by interpretation is met. For the second, let
$\Sigma \in Sign$ and $\Phi \cup \enset{\phi} \subseteq SEN(\Sigma)$.
Assume $\phi \in C_\Sigma(\Phi)$. Then
\begin{equation*}
\alpha_{\Sigma}(\phi)\subseteq C^0_{F(\Sigma)}(\alpha_{\Sigma}(\Phi))
\end{equation*}
because $\pair{\f,\alpha}$ is an interpretation. On the other hand, $I^{\prime}$ being a 
syntactic refinement of $I^0$,
\begin{equation*}
C^0_{F(\Sigma)}(\alpha_{\Sigma}(\Phi)) \subseteq C'_{F(\Sigma)}(\alpha_{\Sigma}(\Phi))
\end{equation*}
Thus,
$\alpha_{\Sigma}(\phi)\subseteq C'_{F(\Sigma)}(\alpha_{\Sigma}(\Phi))$.
%
\end{proof}

Definition \ref{df:rbi} subsumes the corresponding notion introduced in \cite{MMB09b} for
$k$-dimensional deductive systems, because 
every $k$-dimensional deductive system  $\pair{\cal L, \vdash_{\cal L}}$ 
over a countable set of variables $V$, gives rise to a specific $\pi$-institution
$I_{\cal L}=\pair{Sign_{\cal L},Sen_{\cal L},
            (C_{{\cal L}_{\Sigma}})_{\Sigma \in |Sign_{\cal L}|}}$, built in  \cite{Vo02} as follows:
\begin{description}
\item[(i)] $Sign_{\cal L}$ is the one-object category with object $V$. The
           identity morphism is the inclusion 
            $i_{V}:V\rightarrow Fm_{\cal L}(V)$, where $Fm_{\cal L}(V)$ 
            denotes the set of formulas constructed by recursion using variables in $V$ and connectives in
            $\cal L$ in the usual way. Composition $g\comp f$
            is defined by
            $g\comp f= g^{\star}f$, where
            $g^{\star}: Fm_{\cal L}(V)\rightarrow Fm_{\cal L}(V)$
            denotes the substitution uniquely extending  $g$ to $Fm_{\cal L}(V)$.
\item[(ii)] $\fuc{SEN}_{\cal L}:Sign_{\cal S}\rightarrow Set$ maps V
            to $Fm_{\cal L}^{k}(V)$ and $f:V\rightarrow V$ to $Fm_{\cal L}(V)$
            $(f^{\star})^{k}: Fm_{\cal L}^{k}(V)
            \rightarrow Fm_{\cal L}^{k}(V)$.
             It is easy to see that $\fuc{SEN}_{\cal S}$ is indeed a functor.
\item[(iii)] Finally, $C_{\cal L}$     is the standard closure
              operator $C_{V}: \pow(Fm_{\cal L}(V))\rightarrow 
             \pow(Fm_{\cal L}(V))$   associated with $\pair{\cal L, \vdash_{\cal L}}$, i.e.,
             $C_{V}(\Phi)=\enset{\phi \in Fm^{k}_{\cal L}(V): 
             \Phi\vdash_{\cal S}\phi}$ for all $\Phi\subseteq 
              Fm^{k}_{\cal L}(V)$.
          \end{description}

\begin{example}
The $\pi$-institution of modal logic $S5^{G}$ forms a (syntactic) refinement of the one for 
classical propositional calculus ($CPC$). 
Actually, consider
the modal signature
$\Sigma=\enset{\rightarrow,\wedge,\vee,\neg,\top,\bot,$ $\Box}$.
Modal logic $K$ is defined as an extension of $CPC$ by adding the 
axiom $\Box(p\rightarrow q)\rightarrow (\Box p\rightarrow \Box q)$
and the inference rule $\frac{p}{\Box p}$. Logic $S5^{G}$, on
the other hand, enriches the signature of $K$ with the symbol $\Diamond$,
and $K$ itself with the axioms $\Box p\rightarrow p$,
$\Box p\rightarrow \Box\Box p$ and
$\Diamond p\rightarrow\Box\Diamond p$, cf.\ \cite{BP}. Hence,
since the  signature of both systems contains the signature of
$CPC$ and their presentations extend that of $CPC$ with extra
axioms and inference rules, we have
 $CPC\leadsto K$ and 
$CPC\leadsto S5^{G}$ (actually, $CPC \leadsto K \leadsto S5^{G}$).
Hence, through these refinements, one may capture 
more complex, modally expressed requirements introduced along the refinement process.
\end{example}

\noindent


Given an interpretation $\fdec{\tau}{Fm_{\cal L}(V)}{\pow{(Fm_{\cal L'}(V'))}}$ between two deductive systems 
$\pair{\cal L, \vdash_{\cal L}}$  and $\pair{\cal L', \vdash_{\cal L'}}$, let us define $\pair{\f_\tau,\tau}$ as the translation between $\pi$-institu\-tions
$I_{\cal L}$ and $I_{\cal L'}$, where $\f_\tau$ is a functor between single object categories, mapping, at the object level,
$V$ to $V'$. As expected, 
\begin{lemma}
An $l$-deductive system $\pair{\cal L', \vdash_{\cal L'}}$ is an interpretation of a
$k$-deductive system $\pair{\cal L, \vdash_{\cal L}}$ through an interpretation $\tau$, iff $\langle F_\tau, \tau\rangle$ interprets the
$\pi$-institution $I_{\cal L}$ in $I_{\cal L'}$.
\end{lemma}

\begin{proof}
Assume $\pair{\cal L, \vdash_{\cal L}}$ (respectively, $\pair{\cal L', \vdash_{\cal L'}}$) are defined
over a countable set of variables $V$ (respectively, $V'$).
Being an interpretation between deductive systems, $\tau$ is a multifunction
$\fdec{\tau}{Fm_{\cal L}(V)}{\pow{(Fm_{\cal L'}(V'))}}$ such that, for all
$\Gamma\cup\enset{\phi}\subseteq Fm_{\cal L}(V)$,

\begin{equation}\label{ri1}
\Gamma \vdash_{\cal L} \phi\; \; \Leftrightarrow\; \; \tau(\Gamma) \vdash_{\cal L'} \tau(\phi)
\end{equation}
According to the construction of $I_{\cal L}$, detailed above, this is equivalent to
\begin{equation}\label{ri2}
\phi \in C_{V}(\Gamma) \; \; \Leftrightarrow\; \; \tau(\phi) \subseteq  C_{V'}(\tau(\Gamma))
\end{equation}
\end{proof}

\noindent
Hence, it is immediate to check that 
\begin{corollary}
An $l$-deductive system $\pair{\cal L', \vdash_{\cal L'}}$ is a refinement of a 
$k$-deductive system $\pair{\cal L, \vdash_{\cal L}}$ through an interpretation $\tau$, iff the
$\pi$-institution $I_{\cal L'}$  is a refinement of $I_{\cal L}$ through $\langle F_\tau, \tau\rangle$.
\end{corollary}

As  a final remark, note that, in a very precise sense,  Definition \ref{df:rbi} also covers the case of classical institutions. 
Actually, a  $\pi$-institution corresponding to a classical one can always be defined: for each 
signature $\Sigma$ and set of formulas $\Psi$, take $C_{\Sigma}(\Psi)$ as the set of sentences  satisfied  in all models
validating $\Psi$.

\section{The local view}\label{sc:inside}
Having discussed refinement  by interpretation of $\pi$-institutions, we address now the same sort of 
refinement applied to  specifications inside an arbitrary $\pi$-institution. Such is the \emph{local} view.
Given an arbitrary $\pi$-institution $I = \pair{Sign,\fuc{SEN},(C_{\Sigma})_{\Sigma \in |Sign|}}$,
 a basic, or \emph{flat} specification is defined as 
\begin{equation*}
SP \; =\; \pair{\Sigma, \Phi}
\end{equation*}
where $\Sigma \in |Sign|$ and $\Phi \subseteq \fuc{SEN}(\Sigma)$.
Its  meaning  is the closure of $\Phi$, i.e., $C_{\Sigma}(\Phi)$. 
D. Sannella and A. Tarlecki in \cite{ST88}
 define specification over an arbitrary institution along similar lines, but taking, as semantic domain, classes of models
 instead of logical consequence relations.

As expected,  any morphism $\fdec{\sigma}{\Sigma}{\Sigma'}$ in $Sign$ 
 entails a notion of \emph{local} refinement $\leadsto_{\sigma}$ in $I$ given by
 
 \begin{equation}
 \pair{\Sigma, \Phi} \leadsto_{\sigma} \pair{\Sigma', \Phi'}\; \; \text{if}\; \; \sigma(\Phi) \subseteq C_{\Sigma'} (\Phi')
\end{equation}
For $\sigma$ an inclusion, this may be regarded as a form of syntactic refinement.

Specifications may also be connected by interpretations which, again, corres\-pond to multifunctions
preserving and reflecting consequence. Formally,


\begin{definition}\label{df:specref} Let $\pair{\Sigma, \Phi}$ and $\pair{\Sigma', \Phi'}$ be two specifications over a $\pi$-institution 
$I = \pair{Sign,\fuc{SEN},(C_{\Sigma})_{\Sigma \in |Sign|}}$ and 
	$\fdec{i}{\fuc{SEN}(\Sigma)}{\pow{(\fuc{SEN}(\Sigma'))}}$  a multifunction from $\fuc{SEN}(\Sigma)$
	to $\fuc{SEN}(\Sigma^{\prime})$ . Then
	$i$ is a {\em (local) semi-interpretation} of $\pair{\Sigma, \Phi}$ in $\pair{\Sigma', \Phi'}$ 
	if, for all $\phi \in  
	\fuc{SEN}(\Sigma)$,
	\begin{eqnarray}
\phi \in C_{\Sigma}(\Phi)\; \imp\; i(\phi) \subseteq C_{\Sigma'}(\Phi')
	\end{eqnarray}
	It is a {\em (local) interpretation} of $\pair{\Sigma, \Phi}$ in $\pair{\Sigma', \Phi'}$  if, 
	\begin{eqnarray}\label{eq:localint}
	\phi \in C_{\Sigma}(\Phi)\; \dimp\; i(\phi) \subseteq C_{\Sigma'}(\Phi')
	\end{eqnarray}
Finally, we say that $i$ \emph{(locally) interprets}  $\pair{\Sigma, \Phi}$, if there is a specification $\pair{\Sigma^0, \Phi^0}$ 
on which $\pair{\Sigma, \Phi}$ is interpreted by $i$.
	\end{definition}

Adopting  expression \aspas{$\phi$ \emph{is true in} specification $\pair{\Sigma, \Phi}$} to abbreviate the fact that
$\phi \in C_{\Sigma}(\Phi)$, definition \eqref{eq:localint} can be read as \emph{$\phi$ is true in $\pair{\Sigma, \Phi}$  
iff $i(\phi)$ is true in $\pair{\Sigma', \Phi'}$}.

\begin{definition}\label{df:specrbi}
Let $SP= \pair{\Sigma, \Phi}$ be a specification and  $\fdec{i}{\fuc{SEN}(\Sigma)}{\pow{(\fuc{SEN}(\Sigma'))}}$ 
a translation which interprets $SP$.  A specification $SP' = \pair{\Sigma', \Phi'}$ refines $SP$ via local interpretation $i$,
written as $SP \leadsto_i SP'$,
if for all $\phi \in \fuc{SEN}(\Sigma)$, 
\begin{equation}
\phi \in C_{\Sigma}(\Phi)\; \imp\; i(\phi) \subseteq C_{\Sigma'}(\Phi')
\end{equation}
\end{definition}

Given a $\sigma:\Sigma\rightarrow \Sigma'\in Sign$, $\fuc{SEN}(\sigma):\fuc{SEN}(\Sigma)\rightarrow \fuc{SEN}(\Sigma')$ induces a translation that maps each $\phi\in \fuc{SEN}(\Sigma)$ into $\{\fuc{SEN}(\sigma)(\phi)\}$. In the sequel we identify this translation simply with $\fuc{SEN}(\sigma)$.

\begin{definition} A signature morphism $\sigma:\Sigma\rightarrow \Sigma'\in Sign$ is \emph{conservative} if for any $\Phi \subseteq \fuc{SEN}(\Sigma)$, $\fuc{SEN}(\sigma)$ interprets $\langle \Sigma, \Phi\rangle$ in $SP^\sigma=\langle \Sigma', \fuc{SEN}(\sigma)(\Phi) \rangle$.
\end{definition}
Observe that $\fuc{SEN}(\sigma)$ is always a semi-interpretation from $SP$ to $SP^\sigma$. Moreover,  note that 
conservativeness is a stronger notion than that of interpretability. 

\begin{theorem}
Let $\sigma:\Sigma\rightarrow \Sigma'\in Sign$ be a conservative signature morphism, $SP=\langle \Sigma, \Phi\rangle$ a specification over $I$ and $\Phi'\in \fuc{SEN}(\Sigma')$. Then, 
\begin{equation}
	\fuc{SEN}(\sigma)(\Phi)\subseteq C_{\Sigma'}(\Phi') \text{ implies that } SP \leadsto_{\fuc{SEN}(\sigma)} \langle \Sigma', \Phi'\rangle
\end{equation}
\end{theorem}

In practice, new specifications are built from old through application of a number of specification constructors.
As a minimum set we consider operators to join two specifications, to translate one into another,
and to derive one from another going backward along a signature morphism. The following definition characterizes
along these lines a  notion of structured specification in an arbitrary $\pi$-institution.

\begin{definition}\label{df:sc}
Structured specifications over an arbitrary $\pi$-institution $I  = \pair{Sign,\fuc{SEN},(C_{\Sigma})_{\Sigma \in |Sign|}}$ are defined inductively as follows, taking flat specifications
as the base case.
\begin{itemize}
\item For a signature $\Sigma$, the union of specifications $SP_1 = \pair{\Sigma, \Phi_1}$ and
$SP_2 = \pair{\Sigma, \Phi_2}$ is defined as
\begin{equation*}
\fuc{union} (SP_1, SP_2)\; \, =\; \, \pair{\Sigma, \Phi_1 \cup \Phi_2}
\end{equation*}
\item The translation of specification $SP = \pair{\Sigma, \Phi}$ through a morphism  $\sigma:\Sigma\rightarrow \Sigma'$ in $Sign$
is defined as
\begin{equation*}
\fuc{translate}\; SP\;  \fuc{through}\; \sigma \, \; =\; \, \pair{\Sigma', \fuc{SEN}(\sigma)(\Phi)}
\end{equation*}
\item The derivation of a $\Sigma$ specification from   $SP' = \pair{\Sigma', \Phi'}$ 
through a morphism  $\sigma:\Sigma\rightarrow \Sigma'$ in $Sign$
is defined as
\begin{equation*}
\fuc{derive}\; SP'\;  \fuc{through}\; \sigma \, \; =\; \, \pair{\Sigma, \Psi}
\end{equation*}
where $\Psi = \setdef{\psi}{\fuc{SEN}(\sigma)(\psi) \in C_{\Sigma'} (\Phi')}$.
\end{itemize}
 \end{definition}
 

Of course, it is desirable that refinement be preserved by horizontal composition of specifications. In particular,
refinement by interpretation should be preserved by all  specification constructors in Definition \ref{df:sc}.
The result is non trivial. For $\fuc{union}$ we have,
 
 \begin{lemma}
 Let  $\fdec{i}{\fuc{SEN}(\Sigma)}{\pow{(\fuc{SEN}(\Sigma'))}}$  be a local interpretation,
 and  $SP_1 = \pair{\Sigma,\Phi_1}$, $SP_2 = \pair{\Sigma,\Phi_2}$ specifications such that
 $SP_1 \leadsto_i SP'_1$ and $SP_2 \leadsto_i SP'_2$.
 If $i$ interprets $\fuc{union} (SP_1, SP_2)$, then 
 $\fuc{union} (SP_1, SP_2) \leadsto_i \fuc{union} (SP'_1, SP'_2)$.
 \end{lemma}
 \begin{proof}
 For all $\phi \in \fuc{SEN}(\Sigma)$, we reason
  \begin{eqnarray*}
 & &  
 SP_1 \leadsto_i SP'_1\; \; \e\; \;  SP_2 \leadsto_i SP'_2
  \just\dimp{definition}
 \phi \in C_{\Sigma}(\Phi_1) \imp i(\phi)  \subseteq C_{\Sigma'} (\Phi'_1)\; \e\;
\phi \in C_{\Sigma}(\Phi_2) \imp i(\phi)  \subseteq C_{\Sigma'} (\Phi'_2)
 \just\imp{$C_{\Sigma}, C_{\Sigma'}$  monotonic}
 \phi \in (C_{\Sigma}(\Phi_1) \cup C_{\Sigma}(\Phi_2)) \imp i(\phi)  \subseteq (C_{\Sigma'} (\Phi'_1) \cup C_{\Sigma'} (\Phi'_2))
   \just\dimp{definition}
 \fuc{union} (SP_1, SP_2) \leadsto_i \fuc{union} (SP'_1, SP'_2)
 \end{eqnarray*}
  \end{proof}
\noindent

The remaining cases are not straightforward. Actually, achieving compati\-bility  entails the need for imposing
 some non trivial conditions on morphisms.


\section{Towards a refinement calculus}\label{sc:calc}

Having defined refinement by interpretation \emph{across} $\pi$-institutions and \emph{inside} an arbitrary 
$\pi$-institution, this section sketches their interconnections. Our first step is to define how a specification
in an institution $I$ translates to $I'$ along an interpretation. 

\begin{definition}
 Let $\fdec{\rho = \pair{\f, \alpha}}{I}{I'}$ be a translation between $\pi$-institutions $I$ and $I'$ and $SP = \pair{\Sigma, \Phi}$
 a specification in $I$. The translation $\hat{\rho}(SP)$ of $SP$ through $\rho$ is defined by
  \begin{equation}
 \hat{\rho}\, \pair{\Sigma, \Phi} \; =\; \pair{\f(\Sigma), \alpha_{\Sigma}(\Phi)}
\end{equation}
 \end{definition}

Next lemma answers  the following question: is refinement by interpretation over arbitrary $\pi$-institutions
 preserved by the specification constructors?

 \begin{lemma}\label{lm:struct}
 The definition of specification translation is structural over the specification constructors given in definition \ref{df:sc},
 i.e.
 \begin{align*}
  \hat{\rho}\, (\fuc{union} (SP_1, SP_2)) &\; =\; \fuc{union} (\hat{\rho}(SP_1),  \hat{\rho}(SP_2))\\
   \hat{\rho}\, (\fuc{translate}\; SP\;  \fuc{through}\; \sigma) &\; =\; \fuc{translate}\;  \hat{\rho}(SP)\;  \fuc{through}\; \f(\sigma)\\
   \hat{\rho}\, (\fuc{derive}\; SP'\;  \fuc{through}\; \sigma) &\; =\; \fuc{derive}\; \hat{\rho}(SP')\;  \fuc{through}\; \f(\sigma)\\
 \end{align*}

 \end{lemma}
 
 \def\igual{=}
 \begin{proof}
  For the first case  let $SP_1 =\pair{\Sigma_1, \Phi_1}$ and $SP_2 =\pair{\Sigma_2, \Phi_2}$. Then,
 \begin{eqnarray*}
 & &  \hat{\rho}\, (\fuc{union} (SP_1, SP_2)) 
 \just\igual{definition of $\fuc{union}$}
  \hat{\rho}\, \pair{\Sigma, \Phi_1 \cup \Phi_2}
   \just\igual{definition of $\hat{\rho}$}
  \pair{\f(\Sigma), \alpha(\Phi_1 \cup \Phi_2)}
     \just\igual{$\alpha$ is a natural transformation}
  \pair{\f(\Sigma), \alpha(\Phi_1) \cup \alpha(\Phi_2)}
   \just\igual{definition of $\fuc{union}$}
   \fuc{union} (\pair{\f(\Sigma), \alpha(\Phi_1)}, \pair{\f(\Sigma), \alpha(\Phi_2)})
 \just\igual{definition of $\hat{\rho}$}
 \fuc{union} (\hat{\rho}(SP_1),  \hat{\rho}(SP_2))\     
 \end{eqnarray*}
~\\

\noindent 
 Consider now the second case (the third being similar):
  \begin{eqnarray*}
 & &  \hat{\rho}\,  (\fuc{translate}\; SP\;  \fuc{through}\; \sigma)
 \just\igual{definition of $\fuc{translate}$}
\hat{\rho}\, \pair{\Sigma', \sigma(\Phi)}
   \just\igual{definition of $\hat{\rho}$}
 \pair{\f(\Sigma'), \alpha_{\Sigma'}(\sigma(\Phi))}
     \just\igual{$\alpha$ is a natural transformation}
   \pair{\f(\Sigma'), \pow{(\sigma)}( \alpha_{\Sigma}(\Phi))}
 \just\igual{definition of $\fuc{translate}$}
    \fuc{translate}\;  \pair{\f(\Sigma'),  \alpha_{\Sigma}(\Phi)}\;  \fuc{through}\;  \f(\sigma)  
     \just\igual{definition of $\hat{\rho}$}
 \fuc{translate}\;  \hat{\rho}(SP)\;  \fuc{through}\; \f(\sigma)    
 \end{eqnarray*} 
 Note a slight abuse of notation: the extension of $ \hat{\rho}(SP)$ in the conclusion is actually through the powerset extension
 of $\f(\sigma)$.
 \end{proof}
 

%
%
%


\section{Conclusions and related work}\label{sc:conc}

In software development, one often has to resort to a number of different logical systems 
 to capture contrasting aspects of systems' requirements and programming paradigms.
This paper uses $\pi$-institutions to formalize arbitrary logical systems and lifts to such level a 
recently proposed \cite{MMB09a,MMB09b}
approach to refinement based on logical interpretation.

Refinement by interpretation is formulated at both a global (\ie, across $\pi$-institutions) and local (\ie, between specifications
inside an arbitrary $\pi$-institution) level. The paper introduces a notion of structured specification and shows that, at both levels,
refinement by interpretation respects the proposed specification construc\-tors. Actually, the institutional setting
not only makes it possible to go a  step further from \cite{MMB09b} in generalizing
 the concept to arbitrary logics, but also provides a basis to build up a refinement calculus of 
 \aspas{institution-independent}, structured specifications. 
  
  We close the paper with a few remarks on  \emph{refinement by interpretation} in itself
  and some pointers to related work.
 
The idea of relaxing what counts as a valid refinement of an algebraic specification, by
replacing \emph{signature morphisms} by \emph{logic interpretations} is, to the best of our
knowledge, new. The piece of research initiated with \cite{MMB09a} up to the present paper 
was directly inspired by the second and third author's work on algebraic logic as reported, respectively,  in \cite{Mar06}
and \cite{Mad08}, where
the notion of an \emph{interpretation} plays a fundamental role (see, \eg,
\cite{BP89,BP,BR03,proto}) and occurs in different variants.
In particular, the notion of \emph{conservative translation} intensively studied by Feitosa and Ottaviano
\cite{traducoesconservativas} is the closest to our own approach.

 Refinement by interpretation  should also be related to the
extensive work of Maibaum, Sadler and Veloso in the 70's and the 80's, as documented, for example, in \cite{Maibaum1,Maibaum2}.
The authors resort to  interpretations between theories and conservative extensions to define a syntactic notion of
refinement according to which
  a specification $SP'$ refines a specification $SP$ if there is an interpretation of $SP'$ into a conservative extension of $SP$. It is
   shown that these refinements can be vertically composed, therefore  entailing stepwise development.  This notion is, however,
   somehow restrictive since it requires all maps to be conservative, whereas in  program development it is usually
   enough to guarantee that  requirements are preserved by the underlying translation. Moreover, in that approach
    the interpretation edge of a refinement diagram needs to satisfy a number of extra properties.

Related work also appears in \cite{FM93,Vo05} where interpretations between theo\-ries are studied,
as  in the present paper, in the abstract framework of $\pi$-institu\-tions. The first reference is a generalization of the work of Maibaum and his colla\-borators, whereas the second  generalizes to $\pi$-institutions the abstract algebraic logic
 treatment of algebraic semantics on sentential logics. Notions of interpretation between institutions also appear in
 \cite{Bo02} and \cite{Tar95} under the designation of \emph{institution representation}. Differently from the 
 one used in this paper, borrowed from \cite{Vo03}, they are not defined as multifunctions.
  The work of
Jos\'e Meseguer \cite{meseguer} on \emph{general logics}, where a  theory of interpretations between logical systems is developed, should
also be mentioned.

We believe this approach to refinement through logical  interpretation has a real application potential, namely to deal
with specifications spanning through different specification logics. Particularly
deserving to be considered, but still requiring further investigation, are
observational logic \cite{BHK03}, hidden logic
\cite{grigore_thesis,conditional_prof} and behavioral logic \cite{Hen97}.  
As remarked above, the study of refinement preservation by horizontal composition remains a challenge and a 
topic of current research.

Other research topics arise concerns the ways in which \emph{global} and  \emph{local} levels 
interrelate. For example, we are still studying to what extent a local refinement by interpretation of a specification in 
a $\pi$-institution $I$, lifts to another  local refinement  of its translation induced by a global interpretation from $I$ to another
$\pi$-institution $I'$. 

\subsection*{Acknowledgments}
This research was partially supported by Fct (the Portuguese Foundation for Science and Technology) under contract
\texttt{PTDC\-/EIA-\-CCO/108302/2008} --- the \textsc{Mondrian} project, and the \textsc{Cidma} research center.
 M. A. Martins was further supported by 
project  \emph{Nociones de Completud}, reference FFI2009-09345 (MICINN - Spain).
Finally, A. Madeira was also supported by SFRH/BDE/ 33650/2009, a joint PhD grant by FCT and Critical Software S.A., Portugal.

\bibliographystyle{eptcsalpha}
\bibliography{refs}
\end{document}